\documentclass[a4paper]{article}
\usepackage[margin=1.12in]{geometry}
\pdfoutput=1 
\usepackage{microtype}

\usepackage{amsmath}
\usepackage{amssymb}
\usepackage{amsthm}
\theoremstyle{definition}
\newtheorem{definition}{Definition}[section]
\newtheorem{theorem}{Theorem}[section]
\newtheorem{proposition}[theorem]{Proposition}
\newtheorem{lemma}[theorem]{Lemma}

\usepackage{enumerate}
\usepackage{alltt}
\usepackage{graphicx}
\usepackage{url}
\usepackage{latexsym}

\usepackage[colorlinks]{hyperref}

\title{Logical Obstruction to Set Agreement Tasks for Superset-Closed Adversaries}

\author{
    ~~\hfill%
    \parbox[t]{.43\textwidth}{\centering Koki Yagi \hfil\\[1.5ex] \texttt{yagi.koki.27v@st.kyoto-u.ac.jp}}
    ~ \hfil ~
    \parbox[t]{.4\textwidth}{\centering Susumu Nishimura\\[1.5ex] \texttt{susumu@math.kyoto-u.ac.jp}}%
    \hfill
    \\ ~\\
    Dept.\ Math, Graduate School of Science, Kyoto University
    }

\begin{document}

\maketitle





\newcommand{\complexf}[1]{\mathcal{#1}}   
\newcommand{\cplC}{\complexf{C}}   
\newcommand{\cplD}{\complexf{D}}   
\newcommand{\cplI}{\complexf{I}}   
\newcommand{\cplO}{\complexf{O}}   

\newcommand{\keyword}[1]{\emph{#1}}       

\newcommand{\Pfont}[1]{\mathsf{#1}}
\newcommand{\Pinput}[2]{\Pfont{input}_{#1}^{#2}}
\newcommand{\Ptrue}{\Pfont{true}}
\newcommand{\Pfalse}{\Pfont{false}}
\newcommand{\Modalfont}[1]{\mathrm{#1}}
\newcommand{\ModK}[1]{\mathop{\Modalfont{K}_{#1}}}
\newcommand{\ModD}[1]{\mathop{\Modalfont{D}_{#1}}}
\newcommand{\ModC}[1]{\mathop{\Modalfont{C}_{#1}}}
\newcommand{\ModE}[1]{\mathop{\Modalfont{E}_{#1}}}
\newcommand{\ModLang}{\mathcal{L}}                   
\newcommand{\ModLangKCD}{\ModLang_{\Modalfont{KCD}}} 
\newcommand{\AtomProps}{\mathrm{AP}}                 
\newcommand{\AgentSet}{\mathrm{Ag}}                  
\newcommand{\precond}{\mathsf{pre}}                  

\newcommand{\valid}{\models}
\newcommand{\nvalid}{\not\models}

\newcommand{\relCls}[2]{\sim^{#1}_{\ModC{#2}}}  

\newcommand{\kripkeMf}[1]{\mathcal{#1}} 
\newcommand{\smplMf}[1]{\mathcal{#1}} 
\newcommand{\smplInput}[1]{\smplMf{I}^{#1}} 
\newcommand{\actMf}[1]{\mathcal{#1}}  
\newcommand{\amIS}{\actMf{IS}}    
\newcommand{\amBCons}{\actMf{BC}}  
\newcommand{\amPre}[1]{\mathsf{pre}_{#1}} 
\newcommand{\actionRep}[2]{{#1}^{#2}}
\newcommand{\actionIS}[2]{\actionRep{#1}{#2}}
\newcommand{\actionRO}[2]{\actionRep{#1}{#2}}

\newcommand{\amSA}[1]{\actMf{SA}_{#1}}    

\newcommand{\Nat}{\omega}
\newcommand{\Int}{\mathbf{Int}}
\newcommand{\Procs}{\colorSet}
\newcommand{\Value}{\mathit{Value}}

\newcommand{\VertSet}{\mathsf{V}}
\newcommand{\Facet}{\mathsf{F}}
\newcommand{\Div}[1]{{\operatorname{#1}}}
\newcommand{\Chromatic}{\Div{Ch}}

\newcommand{\coloring}{\chi}
\newcommand{\colorSet}{\Pi}

\newcommand{\proj}{\pi}  

\newcommand{\advfont}[1]{\mathcal{#1}}
\newcommand{\advA}{\advfont{A}}
\newcommand{\Round}[1]{\mathcal{R}_{#1}}
\newcommand{\csize}{\mathit{csize}}   

\newcommand{\InputAssig}[1]{\mathit{Input}_{#1}}
\newcommand{\OutputAssig}[1]{\mathit{Output}_{#1}}
\newcommand{\ViewAssig}[1]{\mathit{View}_{#1}}

\newcommand{\minView}[1]{\mathit{minView}(#1)}
\newcommand{\maxView}[1]{\mathit{maxView}(#1)}

\newcommand{\ol}[1]{\overline{#1}}
\newcommand{\ovec}[1]{\vec{#1\mathstrut}}
\newcommand{\oovec}[1]{\overrightarrow{#1\mathstrut}}

\newcommand{\PowerSet}[1]{2^{#1}}
\newcommand{\anglpair}[1]{\langle{#1}\rangle}

\newcommand{\impl}{\Rightarrow}
\newcommand{\revimpl}{\Leftarrow}
\newcommand{\nimpl}{\nRightarrow}
\newcommand{\nrevimpl}{\nLeftarrow}

\newcommand{\defeq}{\triangleq}

\newcommand{\abs}[1]{{\mid}#1{\mid}}
\newcommand{\floor}[1]{\lfloor{#1}\rfloor}
\newcommand{\func}[2]{#1\mathrel{\to}#2}


\begin{abstract}
 In their recent paper (GandALF 2018), 
 Goubault, Ledent, and Rajsbaum 
 provided a formal epistemic model for distributed computing. 
 Their logical model, as an alternative to the well-studied topological model, 
 provides an attractive framework for 
 refuting the solvability of a given distributed task
 by means of \emph{logical obstruction}:
 One just needs to devise a formula, 
 in the formal language of epistemic logic, 
 that describes a discrepancy between the model of computation
 and that of the task. 
 However, few instances of logical obstruction were presented  
 in their paper and specifically logical obstruction 
 to the wait-free 2-set agreement task was left as an open problem.
 Soon later, Nishida affirmatively answered to the problem by providing 
 inductively defined logical obstruction formulas
 to the wait-free $k$-set agreement tasks. 

 The present paper refines Nishida's work and 
 devises logical obstruction formulas to $k$-set agreement tasks for
 superset-closed adversaries, which supersede the wait-free model.
 These instances of logical obstruction formulas
 exemplify that the logical framework can provide yet another feasible
 method for showing impossibility of distributed tasks, though it 
 is currently being confined to one-round distributed protocols.
 The logical method has an advantage
 over the topological method that it enjoys a self-contained, 
 elementary induction proof.
 This is in contrast to topological methods,
 in which sophisticated topological tools, such as Nerve lemma, are 
 often assumed as granted. 
 \end{abstract}


\section{Introduction}
\label{sec:intro}


In the last few decades, 
the topological theory of distributed computing
has been successful in giving a range of 
fundamental insights and results, most notably 
the simplicial complex model of distributed tasks, 
protocols as simplicial subdivisions, and 
the impossibility results on
a significant family of distributed tasks \cite{HerlihyShavit99,Herlihy:DCTopology}. 
%
Though it 
has been recognized since
the earliest work by Saks and Zaharoglou \cite{SaksZaharoglou00}
that the topological model is an interpretation of 
epistemic knowledge held by distributed processes,
the rigorous connection was established 
only very recently by Goubault, Ledent, and Rajsbaum 
in \cite{GoubaultLedentRajsbaum:GandALF18}.
They defined the task solvability in terms of
a multi-agent dynamic epistemic logic (DEL) \cite{DitmarschHoekKooi:DELbook08}, 
establishing the isomorphism between the topological models 
and an appropriate class of logical models of epistemic knowledge.

This new logical model 
yields a formal means to refute task solvability:
One is obliged to find an epistemic logic formula $\psi$, called
\emph{logical obstruction}, 
such that  $\psi$ holds for the model of the task
but does not for the model of the
protocol that is relevant to the system of concern.
However, in \cite{GoubaultLedentRajsbaum:GandALF18}, 
concrete instances of logical obstruction are given only for 
the wait-free binary consensus task but none for others.
Particularly the logical obstruction to
the general wait-free $k$-set agreement task was
left as an open problem.
Soon later, Nishida settled it positively by 
devising concrete logical obstruction formulas
to the wait-free $k$-set agreement tasks \cite{Nishida:Msc20},
where he carried out the proof by elementary inductive argument
on inductively defined formulas.


This paper presents an extension of 
Nishida's result to the computation model 
of superset-closed adversaries \cite{GafniKuznetsov10},
which generalizes the basic wait-free model.
More specifically, we provide logical obstruction formulas to show
that $k$-set agreement tasks for a sufficiently small $k$  
are not solvable under superset-closed adversaries,
by a single round execution of the 
round operator \cite{HerlihyRajsbaum10}. 

The logical obstruction we are to present for the adversarial model 
is a generalization of  the one devised by Nishida 
for the wait-free model. 
It is somewhat surprising that essentially the same logical obstruction,
up to appropriate generalization,  
works well for the models of different levels of complexity.
(See Section~\ref{subsec:obstAdversary} for a discussion of 
how the underlying simplicial structures are different in the two models.)
This is made possible by exploiting 
the notion of \emph{permutation subset},
a combinatorial feature intrinsic to set agreement tasks \cite{Nishida:Msc20}.

The impossibility proofs by Nishida's and ours demonstrate that 
the logical method proposed by by Goubault, Ledent, and Rajsbaum 
\cite{GoubaultLedentRajsbaum:GandALF18} serves as 
yet another feasible method of refuting task solvability,
where the logical obstruction is expressed in
a formal language of epistemic logic.
In the logical framework, one just needs to devise a logical
obstruction formula to the task of concern and is obliged 
to prove that the formula is indeed an obstruction. 
In this paper, we show that a simple elementary inductive argument suffices, 
at least the unsolvability of $k$-set agreement tasks 
for the adversarial model is concerned.

We have to remark that this paper discusses 
task unsolvability is confined to 
single round protocols solely 
and does not concern multiple round protocols, 
i.e., iterated execution of a single round protocol.
This is in contrast to the preceding work 
by Herlihy and Rajsbaum \cite{HerlihyRajsbaum10}
that provides a topological proof for the 
unsolvability of set agreement tasks 
for multiple round protocols. 
We would need a further technical 
development for effectively dealing with multiple round protocols in the 
logical framework and leave this topic for future investigation. 


\paragraph*{Related work.}

Topological methods have been successfully applied 
to show the impossibility of set agreement tasks
for the wait-free model \cite{HerlihyShavit99}
and for the adversary model \cite{HerlihyRajsbaum10}. 
The general proof strategy is to find a `topological' obstruction
that detects the topological inconsistency between the models
of the task and the protocol. In the case of set agreement tasks,  
the unsolvability comes from the fact that
the image of the carrier map \cite{Herlihy:DCTopology} of the task, 
i.e.,  a functional specification of the task that maps 
an input simplex to an output complex, is less connected than 
that of the protocol.
In contrast, in the logical framework proposed by  
Goubault et al.\ \cite{GoubaultLedentRajsbaum:GandALF18}, 
one is asked to devise a logical obstruction formula
in the \emph{product update model},
which encodes the input/output relation of a carrier map 
by a complex obtained by a relevant product construction.
(See Section~\ref{subsubsec:actModelProdUpdate} for the formal definition
of product update model.) 
It is often difficult to translate a topological inconsistency into 
a logical obstruction formula, since 
the complex of a product update model tends to have a fairly different, 
more complicated simplicial structure than that of the original one.



The impossibility results presented in \cite{Nishida:Msc20}
and this paper are demonstrated by a self-contained, elementary proof. 
Once the semantics of the formal language of epistemic logic is learnt,
one is immediately accessible to every detail of the proof.
In contrast, impossibility proofs in topological methods
tend to resort to sophisticated theorems from topology.
For instance, the impossibility proof of set agreement tasks
given in \cite{HerlihyRajsbaum10} critically depends on 
the Nerve lemma, which is by no means an elementary topological result.
(See \cite{Bjorner95} for a combinatorial proof of Nerve lemma
and \cite{Kozlov:CombiAlgTopology} for a proof from the perspective of 
modern algebraic topology.)


It should be noticed that the logic-based solution 
provided in this paper
does not fully substitute for the topological method.
First, as we mentioned earlier, 
our logical obstruction concerns single round protocols only
and is not immediately applicable to multiple round protocols. 
This is because a different logical obstruction formula is required
for each specific epistemic structure that varies at every 
incremental round step. 
In contrast, topological method is more robust to such 
incremental evolution in the underlying structure, 
resorting to a certain topological invariant that is kept intact 
throughout the steps of rounds.
Second, 
the product update model of the logical framework is only able to
relate each input facet (i.e., a simplex of maximal dimension) 
to an output complex (of the same dimension).
This means the logical framework assumes certain
crash-free systems in which no process execution fails, 
as opposed to the usual crash-prone assumption on distributed systems. 
The logical framework, however, still allows us to derive 
impossibility results for set agreement tasks: Under the asynchronous setting, 
faulty processes can be regarded just as processes 
which execute so ``slowly'' that correct processes are 
ignorant of them. \cite{Lendent:PhD20,GafniKuznetsovManolescu14}



%

Naturally, the development of the present paper much owes 
to Nishida's original work \cite{Nishida:Msc20}.
Our contribution is to generalize 
his insight in the construction of logical obstruction formulas
so that it can encompass adversarial models. 
This also contributes to refine Nishida's original proof 
with appropriate level of generalization. 
We hope the present paper would 
help the significant idea in Nishida's paper,
which is of limited accessibility, 
to reach a wider audience of interest.

The rest of the paper is organized as follows.
Section~\ref{sec:DELsemantics} reviews the simplicial semantic model for
epistemic logic and the theorem for task solvability in the model of
dynamic update of knowledge, 
as introduced in \cite{GoubaultLedentRajsbaum:GandALF18}.
In Section~\ref{sec:obstruction}, we reproduce a classical impossibility result
in the DEL framework by means of logical obstruction,
showing that a simple combinatorial argument suffices
for demonstrating the obstruction proof,
without recourse to topological arguments.
In Section~\ref{sec:impossibleSA}, generalizing Nishida's,
we provide logical obstruction formulas
to $k$-set agreement tasks for superset-closed adversaries
and provides an inductive proof on the cardinality
of adversaries.
Section~\ref{sec:conclusion} concludes 
the paper with a summary and a discussion on
directions for future research.


\section{The DEL Model for Distributed Computing}
\label{sec:DELsemantics}

Throughout the paper, 
we assume that a distributed system consists of 
$n+1$ asynchronous processes ($n\geq 0$),
which are given unique ids. 

\subsection{Simplicial topology for distributed computing}
\label{subsec:topologyDC}

A \keyword{simplicial complex} (\keyword{complex} for short) $\cplC$
is a family of finite sets of vertexes, called \keyword{simplexes}, closed under set inclusion, 
that is, $Y\subseteq X$ and $X\in \cplC$ implies $Y\in \cplC$ for any pair of simplexes $X$ and $Y$. 
We say that a complex $\cplD$ is a \keyword{subcomplex} of $\cplC$, if
$\cplD\subseteq\cplC$. 
A simplex $X\in \cplC$ is of \keyword{dimension} $n$ if $\abs{X}=n+1$. 
We call a simplex $X\in \cplC$ a \keyword{facet}, if $X$ is 
a maximal simplex in $\cplC$, i.e., 
$X\subseteq Y$ implies $X=Y$ for any $Y\in\cplC$. 
A complex $\cplC$ is called \keyword{pure} (of dimension $n$), 
if every facet of $\cplC$ has the same dimension $n$. 
We write $\VertSet(\cplC)$ for the set of all vertexes in $\cplC$
and $\Facet(\cplC)$ for the set of all facets in $\cplC$. 

In topological methods for 
distributed computing \cite{Herlihy:DCTopology}, 
a simplex is used for modeling 
a system state of the collection of asynchronous processes,
where each vertex stands for a state of an individual process.
A simplicial complex stands for the set of possible states of 
a distributed system.
In this paper, we are solely concerned with the so-called 
\keyword{colored} distributed tasks. Colored tasks are modeled
by \keyword{chromatic} simplicial complex, 
where the coloring function $\coloring$ assigns a color $\coloring(v)$
for each vertex $v$ so that different vertexes contained in the same simplex
are distinctively colored, that is, 
whenever $u, v\in X$ for a simplex $X$, 
$\coloring(u)=\coloring(v)$ implies $u=v$.
The coloring function models the assignment of
unique ids to individual processes.


Throughout the paper, we are solely concerned with
pure $n$-dimensional chromatic complex,
which we simply call \keyword{complexes}
in the rest of the paper. A chromatic complex is formally denoted by a
pair $\anglpair{\cplC,\coloring}$. Since 
a (pure chromatic) complex $\cplC$
is equally defined by the set of its facets
$\Facet(\cplC)$, so that $\cplC = \{Y\mid Y\subseteq X
\text{ for some }X\in\Facet(\cplC) \}$, 
we may also write a complex as $\anglpair{\Facet(\cplC),\coloring}$,
or even in the abridged notation $\cplC$,
leaving the coloring function $\coloring$ implicit.

A \keyword{simplicial map} $\mu: \VertSet(\cplC)\to\VertSet(\cplD)$,  
where $\cplC$ and $\cplD$ are complexes,
is a vertex map such that $\mu(X)$ is a simplex of $\cplD$
for any simplex $X$ of $\cplC$. 
In addition, we postulate that $\mu$ is \keyword{color-preserving}, that is,
$\coloring(v)=\coloring(\mu(v))$ for every vertex $v\in\VertSet(\cplC)$. 
Hence $\mu$ preserves the dimension of simplexes, in particular,
it maps a facet of $\cplC$ to a facet of $\cplD$.
Furthermore, simplicial map commutes with intersection on simplexes,
namely, $\mu(X\cap Y)=\mu(X)\cap\mu(Y)$.

\subsection{Epistemic logic and its semantics}
\label{subsec:ELsemantics}

The (multi-agent) epistemic logic \cite{DitmarschHoekKooi:DELbook08} is a logic for 
formal reasoning of knowledge of an individual agent or a group of agents.
We assume a set $\AtomProps$ of atomic propositions:
An atomic proposition is a propositional 
symbol that states whether a certain property is qualified.
The following grammar defines the formal language of epistemic logic formulas:
\[
    \varphi::=
    p \mid \varphi\vee\varphi  \mid \varphi\wedge\varphi \mid  \neg\varphi
    \mid \ModK{a}\varphi \mid \ModC{A}\varphi \mid \ModD{A}\varphi,
\]
where $p$ ranges over $\AtomProps$,
$a$ ranges over a finite set $\AgentSet$ of agents, and
$A$ ranges over the powerset $\PowerSet{\AgentSet}$. 
As a usual convention, we write $\Pfalse$
for abbreviation of $p \wedge \neg p$, where $p$ is an arbitrary atomic proposition.

In addition to propositional connectives, epistemic logic formulas 
are equipped with
a few modal operators regarding agents' knowledge:
$\ModK{a}\varphi$ is the `knowledge' operator for a single agent,
reading ``the agent $a$ knows $\varphi$'';
$\ModC{A}\varphi$ and 
$\ModD{A}\varphi$ are operators defined for a group $A$ of agents, called 
`common knowledge' and 'distributed knowledge' operators, respectively.\footnote{%
    There is another operator worth mentioning, called `group knowledge', 
    which is of no use in the present paper
    and is omitted.}

The semantics for epistemic logic formulas
is defined over an appropriate Kripke frame.
A \keyword{Kripke frame} is formally 
a pair $\anglpair{S,\sim}$, where $S$ is
a collection of possible worlds and 
$\sim$ is a family of equivalence relations $\{\sim_a\}_{a\in\AgentSet}$
over $S$.
Each $\sim_a$ is called an \keyword{indistinguishability relation}, meaning that
the worlds $X$ and $Y$ are not distinguishable 
by the agent $a$ iff $X \sim_{a} Y$. 
A \keyword{Kripke model} $\kripkeMf{M}=\anglpair{S,\sim,L}$
is a Kripke frame along with 
a function $L: S\to \PowerSet{\AtomProps}$,
which assigns a set $L(X)$ of atomic propositions that are true
in each world $X\in S$. 

In the following, let $\relCls{}{A}$ denote
the reflexive transitive closure of $\bigcup_{a\in A}\sim_{a}$,
that is, $X \relCls{}{A} Y$ iff
$X = 
X_0 \sim_{a_0} X_1 \sim_{a_1} \cdots \sim_{a_n} X_n 
=Y$
for some $X_0, X_1, \ldots, X_n\in S$ and $a_0,a_1,\ldots, a_n\in A$ 
($n\geq 0$). 
Further, 
we define the relation $\sim_{\ModD{A}}$ by:  
$X \sim_{\ModD{A}} Y$ iff $X\sim_{a}Y$ for every $a\in A$.
Then for each particular world $X$ of Kripke model $\kripkeMf{M}$, 
the semantics, or truth valuation, of an epistemic logic formula $\varphi$
is given by the assertion $\kripkeMf{M}, X \models \varphi$, 
which is defined below by induction on the structure of $\varphi$.
\begin{definition}\label{def:kripkeSemantics}
    Given a Kripke model $\kripkeMf{M}=\anglpair{S,\sim,L}$, 
    the \keyword{satisfaction relation} $\kripkeMf{M}, X \models \varphi$ 
    is defined as follows, by induction on $\varphi$. 
    \begin{align*}
        \kripkeMf{M},X\models p
        & \text{\quad if $p\in L(X)$;}
        \\
        \kripkeMf{M},X\models\varphi\vee\psi
        & \text{\quad if $\kripkeMf{M},X\models\varphi$ or $\kripkeMf{M},X\models\psi$;}
        \\
        \kripkeMf{M},X\models\varphi\wedge\psi
        & \text{\quad if $\kripkeMf{M},X\models\varphi$ and $\kripkeMf{M},X\models\psi$;}
        \\
        \kripkeMf{M},X\models\neg\varphi
        & \text{\quad if $\kripkeMf{M},X\not\models\varphi$;}
        \\
        \kripkeMf{M},X\models\ModK{a}\varphi
        & \text{\quad if $\kripkeMf{M},Y\models\varphi$ for every $Y$ such that $X\sim_{a}Y$; }
        \\
        \kripkeMf{M},X\models\ModC{A}\varphi
        & \text{\quad if $\kripkeMf{M},Y\models\varphi$ for every $Y$ such that $X\relCls{}{A}Y$;}
        \\
        \kripkeMf{M},X\models\ModD{A}\varphi
        & \text{\quad if $\kripkeMf{M},Y\models\varphi$ for every $Y$ such that $X\sim_{\ModD{A}}Y$.}        
    \end{align*}

    We write $\kripkeMf{M}, X \not\models \varphi$ to mean 
    $\kripkeMf{M}, X \models \varphi$ does \emph{not} hold. 
    A formula $\varphi$ is called \keyword{valid} (in $\kripkeMf{M}$), 
    if $\kripkeMf{M}, X\models \varphi$ for every $X\in S$;
    $\varphi$ is called \keyword{invalid},  
    if $\kripkeMf{M}, X \not\models \varphi$ for some $X\in S$.
\end{definition}

We say an epistemic formula \keyword{positive},
if its every subformula of negated form $\neg \psi$ has
no occurrences of
modal operators $\ModK{a}$, $\ModC{A}$, and $\ModD{A}$.
For example, assuming $p,q\in\AtomProps$, 
$\neg q \wedge \ModC{A}\neg p$ is a positive formula, 
while $p\vee \neg (q \vee\ModC{A} p)$ is not,  
because the latter contains a negated subformula $\neg (q \vee\ModC{A} p)$,
whose scope of negation embraces a knowledge operator $\ModC{A}$.

\subsection{Dynamic epistemic logic and its simplicial model}
\label{subsec:simpEpistemicModel}

In \cite{GoubaultLedentRajsbaum:GandALF18}, 
Goubault et al.\ have shown that 
there is a tight correspondence between 
Kripke frame 
and simplicial complex. 
Restricted to 
an appropriate class of 
Kripke models,  
they are indeed the dual of each other
in a suitable category theoretic setting \cite{GoubaultRajsbaum:arXiv17b}.
In short, a simplicial complex and a Kripke model 
are just different representations of the same structure:
colors correspond to agents, facets to possible worlds, and 
adjacency of facets to relation over possible worlds.

In the rest of the paper, 
$\colorSet= \{0,1,\ldots, n\}$ is assumed to be 
the set of unique ids given to $n+1$ individual processes.
We also write $\Value$ to 
denote the set of possible local values held by the processes.
For convenience, we may often denote a colored vertex 
in a simplicial complex by a pair $(a,v) \in \colorSet\times \Value$, 
which is intended to mean a process of id $a$ that holds a local value $v$.
The coloring function is defined by $\chi((a,v))=a$.

\subsubsection{Simplicial model induced by complex}
\label{subsubsec:simplModelCpl}

In this paper, we assume the set of atomic propositions is
given by $\AtomProps= \{\Pinput{a}{v} \mid a\in \colorSet, v\in \mathit{Value}\}$,
where 
each atomic proposition $\Pinput{a}{v}$ is intended to assert that 
the value $v$ is held by the agent $a$ (that is, the process whose id is $a$). 

The Kripke model that is dual to a given complex $\cplC=\anglpair{\cplC,\coloring}$
is induced as below.
\begin{definition}[Simplicial model]\label{def:simplicialModel}
    For any complex $\cplC$, we can induce a Kripke model 
    $\anglpair{\Facet(\cplC),\sim,L}$ such that:
    \begin{itemize}
        \item the set of agents is taken as the set of colors, i.e., 
        $\AgentSet=\colorSet$, 
        \item the set of possible worlds are the set of facets $\Facet(\cplC)$,
        \item the equivalence relation is defined by:
        \[
            X \sim_a Y \text{ iff } a \in \coloring(X\cap Y),
            \text{ where $X,Y\in\Facet(\cplC)$ and $a\in\AgentSet$, and}        
            \]
        \item $L(X)= \{\Pinput{a}{v} \mid (a,v)\in X\}$.
    \end{itemize}

    We call this induced Kripke model a \keyword{simplicial model}.
    In abuse of notation,
    we also write
    $\cplC$ for the simplicial model $\anglpair{\Facet(\cplC),\sim,L}$.
\end{definition}

Although the entire results of this paper 
can be fully worked out without any help of topological intuition, 
it is instructive to see here the topological implication
of epistemic formulas.
Figure~\ref{fig:simplModel} illustrates a complex $\cplC$ of dimension 2,
from which a simplicial model is induced. 
The complex consists of five facets $X_1$, $X_2$, $X_3$, $X_4$, and $X_5$, 
which are the possible worlds of the induced Kripke model.
The three different agents (processes) are distinguished by colors $0$, $1$, and $2$,
and in the figure each vertex colored $0$ (resp., $1$ and $2$)
is depicted in white (resp., blue and red). A vertex denoted by
$w_i$ (resp., $b_i$ and $r_i$) indicates that the
value $i$ is held by the vertex whose color is white (resp., blue and red). 
We assume $L(X)= \{ \Pinput{0}{i}, \Pinput{1}{j}, \Pinput{2}{k} \}$
if and only if $X$ is a facet comprising of vertexes $w_{i}$, $b_{j}$, and $r_{k}$.
For example, $L(X_3)= \{ \Pinput{0}{0}, \Pinput{1}{3}, \Pinput{2}{2} \}$.

In the induced model, it holds that 
$\cplC, X_1 \models \ModK{0} \bigl(\bigvee_{a\in\colorSet} \Pinput{a}{1}\bigr)$,
because there are only two facets related with $X_1$ by  
$\sim_0$ via the vertex $w_2$, 
namely $X_1$ itself and $X_2$, and $\bigvee_{a\in\colorSet} \Pinput{a}{1}$ holds
in both facets. On the other hand,
$\cplC, X_5 \models \ModK{2} \bigl(\bigvee_{a\in\colorSet} \Pinput{a}{1}\bigr)$
does \emph{not} hold, because $X_5 \sim_2 X_3$
but the value $1$ is not held by none of vertexes of $X_3$.
As for distributed knowledge, 
since $\sim_{\ModD{\{0,1\}}}$ relates $X_3$ with solely itself
and $\sim_{\ModD{\{0,2\}}}$ further relates
$X_3$ with $X_4$, 
we have $\cplC,X_3\models \ModD{\{0,1\}} \Pinput{1}{3}$
but $\cplC,X_3\not\models \ModD{\{0,2\}} \Pinput{1}{3}$
because $\cplC,X_4 \not\models \Pinput{1}{3}$.
An assertion on common knowledge $\cplC, X_5\models 
\ModC{\colorSet} \bigl(\bigvee_{a\in\colorSet} \Pinput{a}{2}\bigr)$ holds,
because $\relCls{}{\colorSet}$ relates $X_5$ with 
those facets that belong to the connected component of $X_5$, 
namely all the facets in $\cplC$, and
$\bigvee_{a\in\colorSet} \Pinput{a}{2}$ holds for every facet.

\begin{figure}[t]
    \centering
    \includegraphics[scale=0.33]{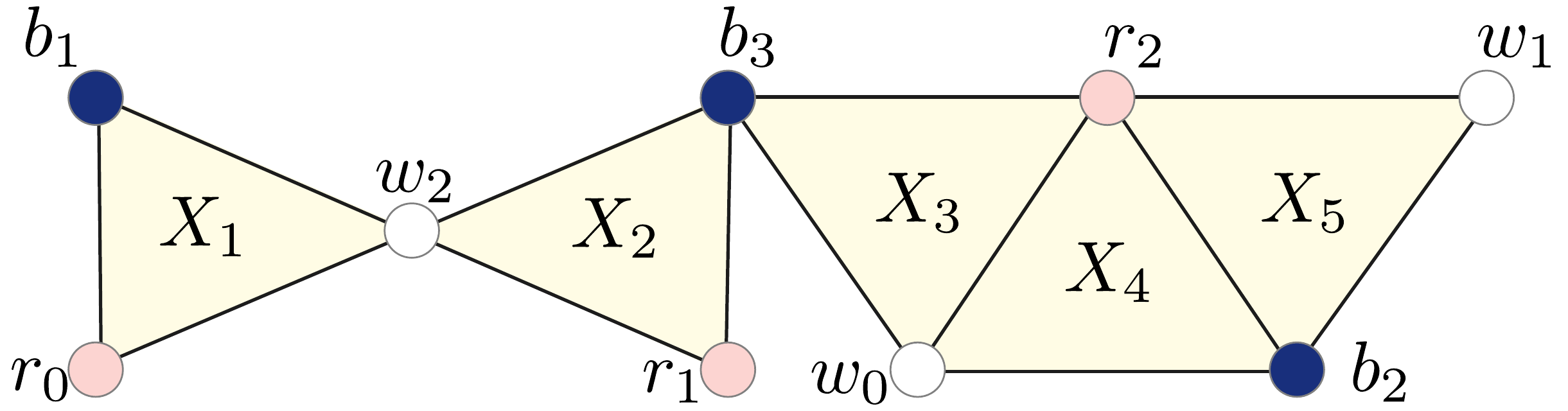}
    \caption{A complex $\cplC$ of dimension 2 that induces
        a simplicial model $\anglpair{\Facet(\cplC),\sim,L}$}
    \label{fig:simplModel}
\end{figure}

As for positive epistemic formulas,
it can be shown that the epistemic knowledge never increases along with 
an appropriate simplicial map.
\begin{theorem}[Knowledge gain\cite{GoubaultLedentRajsbaum:GandALF18}] \label{th:knowledgegain}
    Suppose we are given a pair of simplicial models 
    $\anglpair{\Facet(\cplC),\sim^{\cplC},L^{\cplC}}$ and
    $\anglpair{\Facet(\cplD),\sim^{\cplD},L^{\cplD}}$ that are
    induced from complexes $\cplC$ and $\cplD$, respectively. 
    We call a function $\delta:\cplC\to\cplD$ 
    a \keyword{morphism}, if it  
    is a color-preserving simplicial map from $\VertSet(\cplC)$ to 
    $\VertSet(\cplD)$ such that
    $L^{\cplD}(\delta(X))=L^{\cplC}(X)$ for every $X\in\Facet(\cplC)$.
  
    For any morphism $\delta:\cplC\to\cplD$ and any positive formula $\varphi$,
    $\cplD,\delta(X)\models\varphi$ implies $\cplC,X\models\varphi$, 
    for every facet $X\in \Facet(\cplC)$. 
\end{theorem}

\begin{proof}
    Proof is by induction on $\varphi$. We only examine a few cases that
    have not been demonstrated in \cite{GoubaultLedentRajsbaum:GandALF18}.

    To prove the case of disjunction $\varphi\vee\psi$, 
    suppose $\cplD, \delta(X)\models \varphi\vee\psi$.
    Without loss of generality, 
    we may assume $\cplD, \delta(X)\models \varphi$. By induction hypothesis, 
    $\cplC, X\models \varphi$ and  thus we have  $\cplC, X\models \varphi\vee\psi$.

    For the case of distributed knowledge operator,
    suppose $\cplD, \delta(X)\models \ModD{A}\varphi$.
    In order to show  $\cplC, X \models \ModD{A}\varphi$, 
    we assume $Y\in\Facet(\cplC)$ is an arbitrary facet 
    such that $X \sim_{\ModD{A}} Y$, i.e., $A\subseteq \coloring(X\cap Y)$,
    and show $\cplC, Y \models \varphi$.
    Since $\delta$ is a color-preserving simplicial map,
    $A\subseteq\coloring(X\cap Y)=\coloring(\delta(X\cap Y))=\coloring(\delta(X)\cap\delta(Y))$, which means
    $\delta(X)\sim_{\ModD{A}}\delta(Y)$. By the assumption $\cplD, \delta(X)\models \ModD{A}\varphi$, we have
    $\cplD,\delta(Y)\models \varphi$; hence $\cplC,Y\models \varphi$
    by induction hypothesis.
\end{proof}

From a topological perspective, 
it can be understood that 
the knowledge gain theorem stems from the fact
that, for any morphism $\delta$, 
a positive epistemic formula is less likely to hold in 
$\delta(\cplD)$, which is a more densely connected model than $\cplC$.
Thus we can make use of positive epistemic formulas for detecting 
discrepancy in connectivity of 
simplicial models, as we will demonstrate in later sections.

\subsubsection{The action model and product update}
\label{subsubsec:actModelProdUpdate}

The dynamic epistemic logic (DEL) \cite{DitmarschHoekKooi:DELbook08} is an
epistemic logic with possible updates in the knowledge model.
Following \cite{GoubaultLedentRajsbaum:GandALF18}, 
we make use of \keyword{product update model} generated by
epistemic actions \cite{BaltagMossSolecki16,DitmarschHoekKooi:DELbook08} for modeling 
the change in the knowledge structure incurred by 
an action of information exchange among communicating processes.



To define product update model, let us first 
define the cartesian product of simplicial complexes.
Suppose $\anglpair{\cplC,\coloring_{\cplC}}$ 
and $\anglpair{\cplD,\coloring_{\cplD}}$ are complexes.
For each pair of simplexes $X\in \cplC$ and $Y\in \cplD$
such that $\coloring_{\cplC}(X)=\coloring_{\cplD}(Y)$, 
we define cartesian product $X\times Y$ as a simplex whose each vertex
is a pair of vertexes of matching color from original simplexes, namely,
\begin{align*}
    X\times Y & {} = \{ 
        (u,v) \mid u\in X, v\in Y, \coloring_{\cplC}(u)=\coloring_{\cplD}(v)
        \}. 
\end{align*}
The coloring on vertexes in the cartesian product inherits that on 
the original ones, namely,
$\coloring((u,v))=\coloring_{\cplC}(u)$.
The cartesian product $\cplC\times\cplD$ of complexes
of $\cplC$ and $\cplD$ is then defined by
$\Facet(\cplC\times\cplD)=
\Facet(\cplC)\times \Facet(\cplD)=
\{X\times Y \mid X\in\Facet(\cplC), Y\in\Facet(\cplD) \}$.
%

We can define the pair of projection maps 
$\proj_{\cplC}: \cplC\times \cplD \to \cplC$ and 
$\proj_{\cplD}: \cplC\times \cplD \to \cplD$ 
by $\proj_{\cplC}((u,v))= u$ and
$\proj_{\cplD}((u,v))=v$ for every $(u,v)\in \VertSet(\cplC\times\cplD)$,
respectively.
They are both color-preserving simplicial maps. 

\begin{definition}[Simplicial action model and product update] 
    \label{def:smplActionModel}
    A \keyword{simplicial action model} 
    $\actMf{D}=\anglpair{\Facet(\cplD),\sim^{\cplD},\amPre{}}$
    is an action model, where the possible worlds are facets of some complex $\cplD$
    and $\sim^{\cplD}$ is the family of relations $\{\sim^{\cplD}_a\}_{a\in\AgentSet}$
    induced from $\cplD$, i.e., $X \sim^{\cplD}_a Y$ iff 
    $a\in \coloring(X\cap Y)$, where $a\in\AgentSet$
    and $X,Y\in \Facet(\cplD)$,
    and $\amPre{}$ is a function that assigns 
    an epistemic logic formula, called a \keyword{precondition}, 
    to each facet $X$ of $\cplD$.

    The \keyword{product update model} of 
    an initial simplicial model $\smplMf{C}=\anglpair{\cplC,\sim^{\cplC},L}$
    by an action model $\actMf{D}=\anglpair{\Facet(\cplD),\sim^{\cplD},\amPre{}}$
    is a simplicial model $\smplMf{C}[\actMf{D}]=\anglpair{\smplMf{C}[\actMf{D}],\sim,L'}$,
    where  
    $\anglpair{\smplMf{C}[\actMf{D}],\sim}$
    is the Kripke frame induced from 
    the complex
    $\smplMf{C}[\actMf{D}] 
    =\{Z\in \cplC\times\cplD
    \mid \smplMf{C},\proj_{\cplC}(Z)\models \amPre{}(\proj_{\cplD}(Z))\}$
    and $L'(Z) = L(\proj_{\cplC}(Z))$
    for every $Z\in \cplC\times\cplD$.
\end{definition}

The product update model provides an alternative way for specifying 
distributed tasks. In the topological setting, a task is modeled by 
a carrier map $\cplC \to 2^\cplD$, which 
maps each input simplex in $\cplC$ to a subcomplex $\cplD$ consisting of 
possible output simplexes \cite{Herlihy:DCTopology}. 
The product update model $\smplMf{C}[\actMf{D}]$ encodes the carrier map
as a subset of the cartesian product $\cplC \times \cplD$, 
where $\smplMf{C}, X\models \amPre{}(Y)$
determines how an input simplex $X\in\cplC$ maps to an output simplex $Y\in\cplD$.

\subsection{Task solvability in DEL semantics}
\label{subsec:taskDEL}

Let us consider the task solvability for distributed system of
$n+1$ asynchronous processes, whose colors (process ids) are given 
by the set $\colorSet=\{0,1,\dots,n\}$.
We assume that, at the initial configuration, 
every process is given an arbitrary input taken from a finite 
set $\mathit{inp}$ of values
and that the system solves a task by means of a particular
protocol, a certain distributed procedure for 
exchanging values among processes, so that each process in the system
decides the final output that conforms to the requirement by the task. 

The \keyword{initial simplicial model} 
$\anglpair{\smplInput{\mathit{inp}},\sim^{\smplInput{\mathit{inp}}},L}$
is the simplicial model induced from
the complex $\anglpair{\smplInput{\mathit{inp}},\coloring_{\smplInput{\mathit{inp}}}}$
such that 
\begin{itemize}
    \item The facets of 
    $\anglpair{\smplInput{\mathit{inp}},\coloring_{\smplInput{\mathit{inp}}}}$
    comprises of simplexes 
    $\{ (0,v_0), (1,v_1), \ldots, (n,v_n)\}$
    of dimension $n$, where $v_0,v_1,\ldots,v_n\in \mathit{inp}$ 
    and the color of each vertex $(a,v)$ is determined by
    $\coloring_{\smplInput{\mathit{inp}}}((a,v))=a$ and
    \item $L(X) = \{ \Pinput{a}{v} \mid (a,v)\in X\}$ 
    for every facet $X$ of $\smplInput{\mathit{inp}}$.
\end{itemize} 

A protocol is modeled by a simplicial action model 
$\anglpair{\cplC,\sim^{\cplC},\amPre{\cplC}}$, called a \keyword{communicative action model},
and a task is modeled by a simplicial action model  
$\anglpair{\actMf{T},\sim^{\actMf{T}},\amPre{\actMf{T}}}$,
where $\cplC$ (resp., $\actMf{T}$) models 
the possible results of the execution of the protocol
(resp., the possible outputs of the task)
with $\amPre{\cplC}$ (resp., $\amPre{\actMf{T}}$) 
relating facets of $\smplInput{\mathit{inp}}$
with those of $\cplC$ (resp., $\actMf{T}$) appropriately.

\begin{definition}[task solvability]\label{def:taskSolvability}
    Let $\smplInput{\mathit{inp}}$ be an initial simplicial model and 
    $\cplC$ and $\actMf{T}$ be action models for a protocol and a task, respectively,
    as defined above. 
    Then a task $\smplInput{\mathit{inp}}[\actMf{T}]$ is solvable by 
    $\smplInput{\mathit{inp}}[\cplC]$, if 
    there exists a morphism $\delta:\smplInput{\mathit{inp}}[\cplC]
    \to \smplInput{\mathit{inp}}[\actMf{T}]$
    such that $\proj_{\smplInput{\mathit{inp}}}
    = \proj_{\smplInput{\mathit{inp}}}\circ \delta$.
\end{definition}

This definition of task solvability by product update 
conforms to the topological definition based on carrier maps \cite{GoubaultLedentRajsbaum:GandALF18}.
(The cartesian product of complexes is indeed a categorical product,
whose universality is equivalent to the topological definition. 
\cite{GoubaultRajsbaum:arXiv17b})



In the rest of this paper, 
we are solely concerned with \keyword{uniform} communicative action models,
as for product update models of protocols. 
A communicative action model 
$\anglpair{\cplC,\sim^{\cplC},\amPre{\cplC}}$ is
called uniform, 
if each of its action point (i.e., a facet of $\cplC$) 
is uniquely denoted by $\actionRep{X}{r}$, which is
a pair of facet $X\in\Facet(\smplInput{\mathit{inp}})$
and an index $r\in J$, where $J$ is a fixed finite index set.
The precondition is defined by 
$\amPre{\cplC}(\actionRep{X}{r})= \bigwedge_{a\in\colorSet}\Pinput{a}{v_a}$
for each $X= \{(0,v_0),\ldots,(n,v_n)\}$.

The product update model 
$\smplInput{\mathit{inp}}[\cplC]$
with a uniform communicative action $\cplC$ 
is intended to model a protocol 
that produces a uniform output for each input, independent to 
the initial input values:  
A pair of facets $X$ and $Y$ generate
output subcomplexes $\{\actionRep{X}{r}\mid r\in J\}$ 
and $\{\actionRep{Y}{r}\mid r\in J\}$, respectively, 
of different output values but of isomorphic combinatorial structure.
In this way, 
every facet of the product update model $\smplInput{\mathit{inp}}[\cplC]$
is given as a product $X\times \actionRep{X}{r}$. 
For brevity, we write $\actionRep{X}{r}$ to 
denote a facet of $\smplInput{\mathit{inp}}[\cplC]$,
suppressing the duplicate occurrence of $X$.


\section{Logical Obstruction to Wait-free Binary Consensus Task}
\label{sec:obstruction}

By the task solvability definition in DEL (Definition~\ref{def:taskSolvability})
and the knowledge gain theorem (Theorem~\ref{th:knowledgegain}), 
we obtain a logical method for refuting the solvability
of distributed task.
\begin{theorem}[\cite{GoubaultLedentRajsbaum:GandALF18}]\label{th:obstruction}
    Let $\smplMf{I}$ be an input simplicial model,     
    $\actMf{C}$ be a communicative action model for a protocol, and
    $\actMf{T}$ be a task action model.
    Then, the task $\smplMf{I}[\actMf{T}]$ is 
    not solvable by the protocol $\smplMf{I}[\actMf{C}]$ 
    if there exists a positive epistemic logic formula $\varphi$ such that
    $\smplMf{I}[\actMf{T}] \valid \varphi$ but 
    $\smplMf{I}[\actMf{C}]\nvalid \varphi$;  That is,
    $\varphi$ holds in every world
    of $\smplMf{I}[\actMf{T}]$ but
    $\varphi$ is falsified for some world in $\smplMf{I}[\actMf{C}]$.
\end{theorem}

The positive epistemic formula $\varphi$ is called \keyword{logical obstruction}
to the solvability of the task $\smplMf{I}[\actMf{T}]$ 
for the protocol $\smplMf{I}[\actMf{C}]$. 

As a demonstration of impossibility proof by means of logical obstruction, 
we examine the wait-free distributed binary consensus task 
(i.e., 1-set agreement task with exactly 2 possible initial input values).  
The well-known fact that the binary consensus task is not wait-free solvable 
in the read-write shared-memory model 
with two or more processes has been shown by varying methods 
such as valency argument \cite{FischerLynchPaterson85,Herlihy91}
and combinatorial topology \cite{HerlihyShavit99,AttiyaRajsbaum02,Herlihy:DCTopology}.

Here we demonstrate that this well-known result is also reproduced 
in the logical framework of DEL.
As Goubault et al.\ \cite{GoubaultLedentRajsbaum:GandALF18}
have already shown, the logical obstruction to the binary consensus task 
can be argued based on a topological intuition: 
The input to the task has a connected combinatorial structure, 
but the output does not.
Here we give a concrete logical obstruction formula in the formal language of
epistemic logic and thereby 
demonstrate how a purely combinatorial argument works, 
without recourse to topological infrastructure.

We assume $\colorSet = \{0,1,\ldots,n\}$ ($n\geq 1$)
and the set of binary inputs is $\{0,1\}$.
To show the unsolvability of binary consensus task in the wait-free 
read-write shared memory model, 
we consider the full-information protocol of the execution model, namely 
the immediate snapshot protocol \cite{BorowskyGafni:PODC93,BorowskyGafni:STOC93}. 
In the immediate snapshot protocol, 
the $n+1$ processes are arbitrarily divided into an \keyword{ordered set partition}
$S_1|S_2|\cdots|S_m$, which is an ordered sequence of 
nonempty, pairwise disjoint subsets $S_i$'s
such that $\colorSet=\bigcup_{i=1}^m S_i$. 
Each partition $S_i$ represents a concurrency class and  
the snapshots are taken per each partition one after another in the order,
where every process in the same concurrency class $S_i$ simultaneously writes its local value 
to the shared memory and then collects the \emph{view}, i.e., 
the set of the values that have been written so far. 
As a sequel, every process that belongs to a concurrency class $S_k$
obtains the view  $\{(a,v_a)\mid  a\in \bigcup_{i=1}^k S_i\}$,
where $v_a$ is the initial input value to process $a$.

\begin{figure}
    \centering

    \includegraphics[scale=.33]{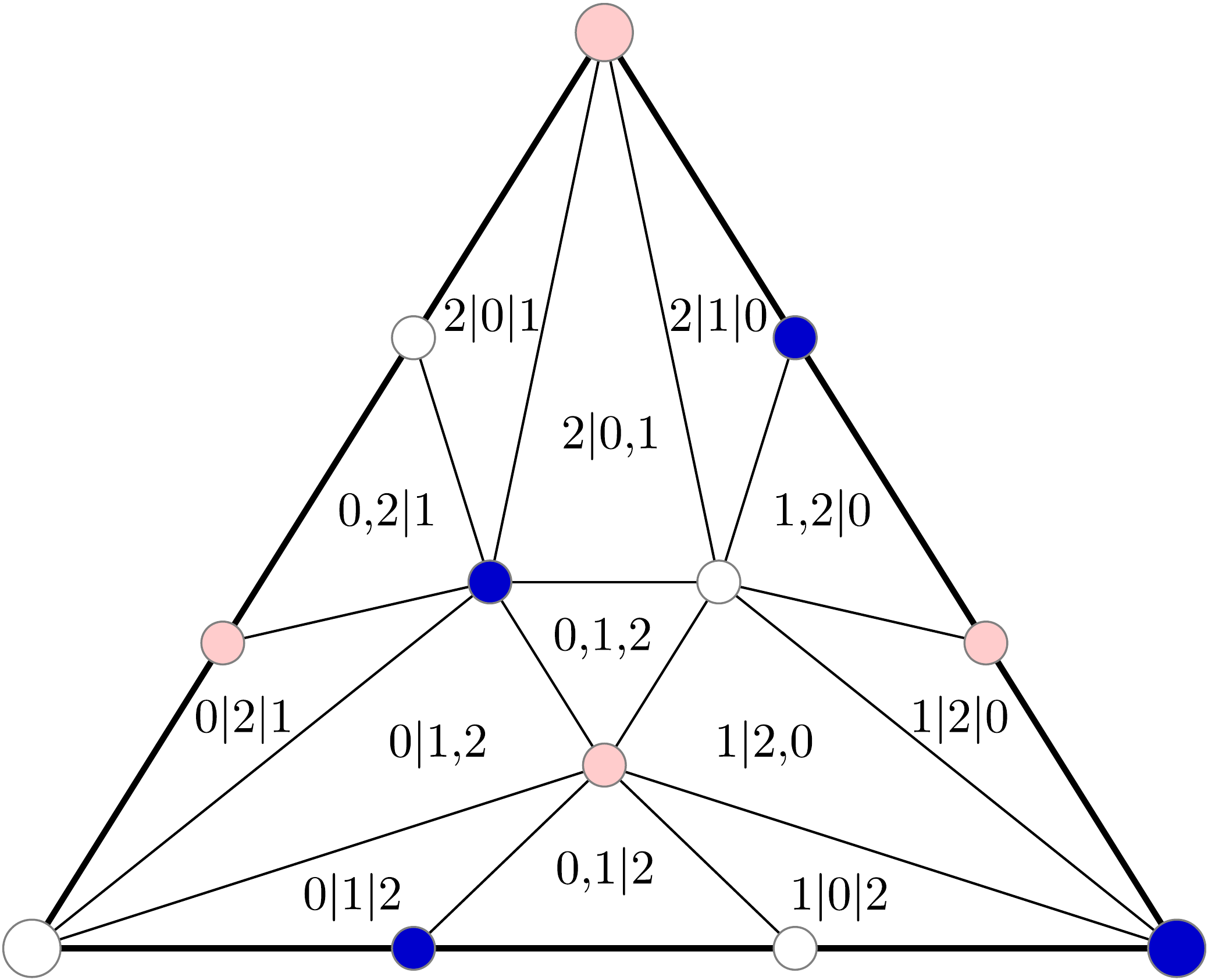}

    \caption{Standard chromatic subdivision for a simplex of dimension 2 ($n=3$),
    with 
    each facet being indexed by an ordered set partition}
    \label{fig:orderedSetPartition}
\end{figure}

It is well-known that a single round execution of the immediate snapshot protocol 
corresponds to a subdivision of the input complex, a.k.a.\ 
the standard chromatic subdivision \cite{HerlihyShavit99,Kozlov12}.
(See Figure~\ref{fig:orderedSetPartition} for a standard chromatic subdivision
and its correspondence to ordered set partitions.)
We write $\Chromatic\,\cplC$ for the standard chromatic subdivision of a complex
$\anglpair{\cplC,\coloring}$. 
We designate a vertex of $\Chromatic\,\cplC$ by a pair $(a,\mathit{view}_a)$, where 
$a$ is the color of the vertex and $\mathit{view}_a$ is the snapshot view of the process $a$.
A facet of $\Chromatic\,\cplC$ is a $\{(a,\mathit{view}_a) \mid a\in\colorSet\}$,
where the snapshot views $\mathit{view}_a$
are derived from a single common ordered set partition. 

The immediate snapshot protocol is modeled 
by the product update  $\smplInput{\{0,1\}}[\amIS]$, 
where $\langle\Facet(\amIS),\sim^{\amIS},$
$\amPre{\amIS}\rangle$
is the uniform communicative action model such that:
\begin{itemize}
    \item The action points are facets of the standard chromatic subdivision
    $\amIS=\Chromatic\,\smplInput{\{0,1\}}$, where     
    we denote each facet by $\actionIS{X}{S_1|S_2|\cdots|S_m}$
    with $X\in \Facet(\smplInput{\{0,1\}})$ and 
    $S_1|S_2|\cdots|S_m$ being an ordered set partition of $\colorSet$;
    \item The preconditions are given as in
    the definition of uniform action model in Section~\ref{subsec:taskDEL}.
\end{itemize}

In what follows, for each facet $X^{S_1|\cdots|S_m}\in \Facet(\smplInput{\{0,1\}})$,
we write 
$\ViewAssig{X^{S_1|\cdots|S_m}}(a)$ for 
$\bigl\{ v\in X \mid
a\in S_k, \coloring(v)\in \bigcup_{i=1}^k S_i \bigr\}$,
which denotes the corresponding view of process $a$.


The binary consensus task is modeled by the product update $\smplInput{\{0,1\}}[\amBCons]$, 
where 
$\langle\Facet(\amBCons),\sim$, $\amPre{\amBCons}\rangle$ is 
the action model such that:
\begin{itemize}
    \item The underlying complex 
    $\amBCons$ consists of solely two facets $\hat{0}=\{(a,0) \mid a\in\colorSet\}$ 
    and $\hat{1} = \{(a,1) \mid a\in\colorSet\}$;

    \item The preconditions are defined by
    $\amPre{\amBCons}(\hat{0})=\bigvee_{a\in\colorSet}\Pinput{a}{0}$ 
    and 
    $\amPre{\amBCons}(\hat{1})=\bigvee_{a\in\colorSet}\Pinput{a}{1}$.    
\end{itemize}

The action model $\amBCons$ defines the binary consensus task,
whose output must be a unanimous agreement to 
one of input values, either $0$ or $1$.

\begin{theorem} \label{th:BCimpossibility}
    There is a logical obstruction formula to  
    the binary consensus task $\smplInput{\{0,1\}}[\amBCons]$ 
    for the protocol $\smplInput{\{0,1\}}[\amIS]$.
    This refutes the solvability of 
    the binary consensus task by the single round immediate snapshot protocol.
\end{theorem}

\begin{proof}
    We will show the following formula works as logical obstruction.
    \begin{align*}
        \Psi = {} &  \neg \Biggl(\bigwedge_{a\in\colorSet}\Pinput{a}{0}\Biggr) 
        \vee \ModC{\Pi}\Biggl(\bigvee_{a\in\colorSet}\Pinput{a}{0}\Biggr)
    \end{align*}
    We are obliged to show
    that $\smplInput{\{0,1\}}[\amBCons]\models \Psi$ 
    but $\smplInput{\{0,1\}}[\amIS]\not\models \Psi$.

    We first show $\smplInput{\{0,1\}}[\amIS]\nvalid \Psi$.
    Let us write 
    $Y= \{(0,0),(1,0),\ldots,(n,0)\}$, 
    $W= \{(0,0),(1,1),\ldots,$
    $(n,1)\}$, and 
    $Z= \{(0,1),(1,1),\ldots,(n,1)\}$, which are the facets of  
    $\smplInput{\{0,1\}}$. 
    It suffices to disprove $\smplInput{\{0,1\}}[\amIS], X_1 \models \Psi$,
    where 
    $X_1=\actionIS{Y}{0|1,\ldots,n}$.    
    Obviously, $X_1 \not\models \neg\bigwedge_{a\in\colorSet}\Pinput{a}{0}$. 
    Let us consider facets
    $X_2=\actionIS{W}{0|1, \ldots, n}$, 
    $X_3=\actionIS{W}{0,1, \ldots, n}$, 
    $X_4=\actionIS{W}{1, \ldots, n|0}$,  and
    $X_5=\actionIS{Z}{1, \ldots, n|0}$, which are related 
    in $\smplInput{\{0,1\}}[\amIS]$ as follows.
    (See Figure~\ref{fig:ISmodel} for the corresponding adjacency of facets, for $n=3$.)
    \begin{itemize}
        \item $X_1 \sim^{\smplInput{\{0,1\}}[\amIS]}_0 X_2$, because $\ViewAssig{X_1}(0)=\ViewAssig{X_2}(0)=\{ (0,0)\}$;
        \item $X_2 \sim^{\smplInput{\{0,1\}}[\amIS]}_n X_3$, because $\ViewAssig{X_2}(n)=\ViewAssig{X_3}(n)=\{ (0,0),(1,1),\ldots,(n,1)\}$;
        \item $X_3 \sim^{\smplInput{\{0,1\}}[\amIS]}_0 X_4$, because $\ViewAssig{X_3}(0)=\ViewAssig{X_4}(0)=\{ (0,0),(1,1),\ldots,(n,1)\}$;
        \item $X_4 \sim^{\smplInput{\{0,1\}}[\amIS]}_n X_5$, because $\ViewAssig{X_4}(0)=\ViewAssig{X_5}(0)=\{ (1,1),\ldots,(n,1)\}$.
    \end{itemize} 
    These entail that $X_1 \relCls{\smplInput{\{0,1\}}[\amIS]}{\colorSet} X_5$. 
    Since $\smplInput{\{0,1\}}[\amIS],X_5 \not\models \bigvee_{a\in\colorSet}\Pinput{a}{0}$, 
    we conclude that $\smplInput{\{0,1\}}[\amIS],$
    $X_1 \not\models \Psi$.

    \begin{figure}[t] 
        \centering
        \includegraphics[scale=.42]{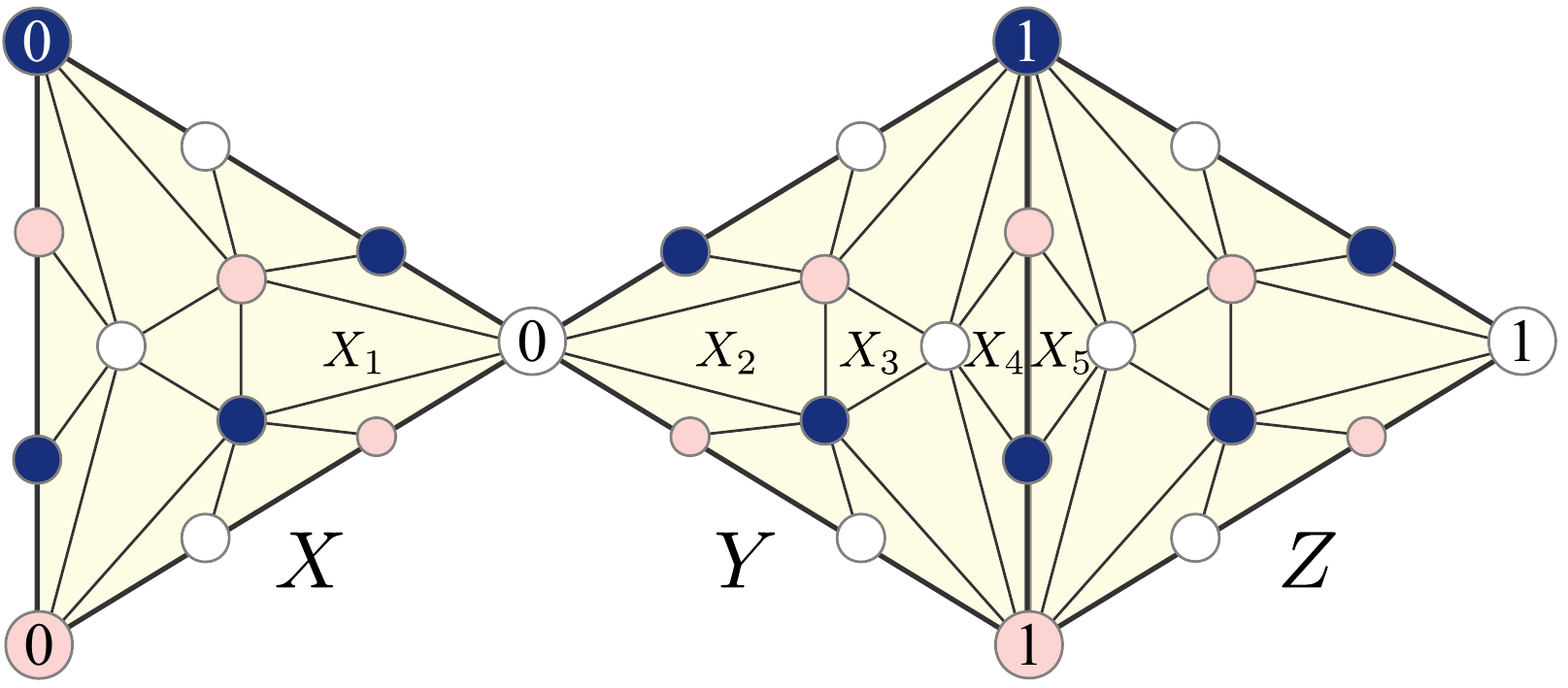}
        \caption{A part of the simplicial model $\smplInput{\{0,1\}}[\amIS]$ for immediate snapshot ($n=3$),
        where $X_i$'s are subdivisions of larger facets $X$, $Y$, and $Z$.}
        \label{fig:ISmodel}
    \end{figure}

    We then prove $\smplInput{\{0,1\}}[\amBCons]\valid \Psi$.     
    Suppose, in contradiction, 
    $\smplInput{\{0,1\}}[\amBCons],X_0 \nvalid \Psi$ for some facet $X_0$ in 
    $\smplInput{\{0,1\}}[\amBCons]$.
    $X_0$ must be a product $Y_0\times \hat{0}$ for some
    $Y_0\in \Facet(\smplInput{\{0,1\}})$,
    because $\amPre{}(\hat{1})=\bigvee_{a\in\colorSet}\Pinput{a}{1}$
    means $\smplInput{\{0,1\}}[\amBCons],X_0 \models \neg \Pinput{a}{0}$ for some $a\in\colorSet$.
    Furthermore, it must be 
    $Y_0= \{(0,0), (1,0), \ldots, (n,0)\}$.
    %
    Since $\smplInput{\{0,1\}}[\amBCons],X_0 \not\models \ModC{\Pi}(\bigvee_{a\in\colorSet}\Pinput{a}{0})$,
    there must exist $X_1, \ldots, X_m \in \smplInput{\{0,1\}}[\amBCons]$
    and $a_1,\ldots,a_m\in \colorSet$
    ($m>0$)
    such that $X_0 \sim^{\smplInput{\{0,1\}}[\amBCons]}_{a_1} X_1 \cdots
    \sim^{\smplInput{\{0,1\}}[\amBCons]}_{a_m} X_m$  and 
    $\smplInput{\{0,1\}}[\amBCons],X_m\not\models \bigvee_{a\in\colorSet}\Pinput{a}{0}$. 
    Hence 
    it must be $X_m=Y_m\times \hat{1}$ and 
    $Y_m = \{(0,1),(1,1),\ldots, (n,1)\}$
    and therefore 
    $X_k=Y_k \times \hat{0}$ and 
    $X_{k+1}= Y_{k+1}\times \hat{1}$ for some $k$. However, 
    $X_k \sim^{\smplInput{\{0,1\}}[\amBCons]}_{a_{k+1}} X_{k+1}$ implies that
    the $X_k$ and $X_{k+1}$ share a common vertex 
    $(v,u) \in X_k \cap X_{k+1}$,
    which leads to a contradiction $0=u=1$.

\end{proof}


\section{Impossibility of General Set Agreement Tasks}
\label{sec:impossibleSA}

This section applies the logical method introduced in previous sections to
set agreement tasks. 
A $k$-set agreement is a distributed computation such that, 
when each process is initially given
an arbitrary input value, processes decide at most $k$
different values taken from the initial inputs. 
As before, we assume $n+1$ processes, whose unique process ids are given 
by the set $\colorSet=\{ 0,1,\ldots,n \}$ of colors.
Without loss of generality, we may assume that the input values 
are taken from $\colorSet=\{ 0,1,\ldots,n \}$, 
renaming the input values appropriately.

The initial simplicial model is then given by $\smplInput{\colorSet}$,
as defined in Section~\ref{subsec:taskDEL}.
Throughout this section, for brevity, we write $\smplInput{}$ for
$\smplInput{\colorSet}$.
The simplicial action model for $k$-set agreement is
given by $\anglpair{F(\amSA{k}),\sim^{\amSA{k}},\amPre{}}$ such that:
\begin{itemize}
    \item the action points are the facets of a complex $\amSA{k}$ that are given by
    \[ 
        F(\amSA{k})= \{ \anglpair{d_0, \ldots, d_n} \mid 
        d_0,\ldots,d_n\in \colorSet, \abs{\{d_0, \ldots, d_n\}}\leq k\},  \]
    where we write $\anglpair{d_0, \ldots, d_n}$ to denote 
    a facet $\{(0,d_0), \ldots, (n,d_n)\}$, whose
    vertexes are colored by $\coloring(a,d)=a$; 
    \item the precondition is given by 
    $\amPre{}(\anglpair{d_0, \ldots, d_n})=\displaystyle
    \bigwedge_{a\in\colorSet}\Bigl(\bigvee_{a'\in\colorSet} \Pinput{a'}{d_a}\Bigr)$.
\end{itemize}
In the following, we may occasionally use 
a vector notation $\vec{d}$ to denote $\anglpair{d_0, \ldots, d_n}$.

The $k$-set agreement task is then modeled by 
the product update $\smplMf{I}[\amSA{k}]$, 
whose worlds are comprised of cartesian products $X\times\anglpair{d_0, \ldots, d_n}$
of facets $X\in\Facet(\smplInput{})$ and 
$\anglpair{d_0, \ldots, d_n}\in\Facet(\amSA{k})$.
Each facet $X\times\anglpair{d_0, \ldots, d_n}$
induces a pair of functions 
$\InputAssig{X\times\anglpair{d_0, \ldots, d_n}}:\colorSet\to\colorSet$
and $\OutputAssig{X\times\anglpair{d_0, \ldots, d_n}}:\colorSet\to\colorSet$
that are defined by 
$\InputAssig{X\times\anglpair{d_0, \ldots, d_n}}(a)=v$ iff $(a,v)\in X$
and also 
$\OutputAssig{X\times\anglpair{d_0, \ldots, d_n}}(a)=d_a$, respectively. 
Then, $X\times\vec{d} \sim_a^{\smplMf{I}[\amSA{k}]}Y\times\vec{d'}$
holds if and only if 
$\InputAssig{X\times\vec{d}}(a)=\InputAssig{Y\times\vec{d'}}(a)$ 
and 
$\OutputAssig{X\times\vec{d}}(a)=\OutputAssig{Y\times\vec{d'}}(a)$. 
Also, $\abs{\OutputAssig{X}(\colorSet)}\leq k$ follows from the definition of 
$\amSA{k}$ and  
$\InputAssig{X\times\vec{d}}(\colorSet)
\supseteq \OutputAssig{X\times\vec{d}}(\colorSet)$
from the precondition.

\subsection{Logical obstruction to wait-free set agreement}
\label{subsec:obstructionWF}

In \cite{Nishida:Msc20}, Nishida presented
inductively defined epistemic formulas as
logical obstruction to the wait-free $k$-set agreement tasks. 
The primary insight in his impossibility proof is that 
each particular execution of a $k$-set agreement task is associated with a 
\keyword{permutation subset}, whose formal definition is given below.

\begin{definition}[Permutation subset]
    We say a function $g$ a \keyword{permutation} of a set $S$, if 
    $g$ is a bijection over $S$. 
    We call a nonempty subset $A$ of a finite set $U$ is 
    a \keyword{permutation subset} for a function $f: U\to U$, 
    if $f(A)=A$, i.e., $A$ is a fixed point of $f$. 
    \unskip\footnote{%
    Nishida called such a fixed point a \emph{family of cycles} \cite{Nishida:Msc20}, 
    because of the well-known fact that every permutation is decomposed into one or more cycles.}
\end{definition}

\begin{lemma}\label{lem:existPermsubset}
    Let $f: U\to U$ be a function over a nonempty finite set $U$. 
    Then there exists a permutation subset $A$ for $f$.
\end{lemma}
\begin{proof}
    Trivially, $f$ is a monotonic function over $2^{U}$, that is, 
    $A\subseteq A'$ implies $f(A)\subseteq f(A')$. 
    By the monotonicity, it hollows that $f^i(U)\supseteq f^{i+1}(U)\neq\emptyset$ 
    from $U\supseteq f(U)$ by induction.  
    Since $U$ is finite, there exists $m$ such that $f^i(U)= f^{i+1}(U)\neq\emptyset$
    for every $i\geq m$. This implies that the set $A= \bigcap_{i=0}^{\infty} f^{i}(U)$
    is a nonempty fixed point, i.e., $f(A)=A$. 
    $A$ is indeed the greatest fixed point of $f$. 
\end{proof}

As an immediate consequence of this lemma, 
$\OutputAssig{X}:\colorSet\to\colorSet$ has a permutation subset 
for any facet $X$ of $\cplI[\amSA{k}]$.

Below we give Nishida's logical obstruction $\Phi$
to the wait-free $k$-set agreement task:
\[
    \Phi ~ \defeq ~
    \bigvee_{a=0}^{n}\neg\Pinput{a}{a} ~
    \vee ~
    \bigvee_{m=1}^{k}\bigvee_{\substack{A\subseteq\colorSet \\ \abs{A}=m}} \Psi^{(m)}_{A},
\]
where $\Psi^{(m)}_{A}$'s are a class of formulas such that 
$A\subseteq\colorSet$ and $\abs{A}=m$ that is defined inductively as follows:
\[
    \Psi^{(m)}_{A} ~ \defeq ~
    \ModD{A}\left(
        \bigvee_{a\in\colorSet\setminus A}\!\neg\Pinput{a}{a} ~
        \vee\bigvee_{a\in\colorSet\setminus A}\ModK{a}
        \Biggl(\bigvee_{j\in A}\bigvee_{a'\in\colorSet}\Pinput{a'}{j}\Biggr)
        \vee\bigvee_{i=m+1}^{n}\bigvee_{\substack{B\subseteq\colorSet\setminus A \\ \abs{B}=i-m}}
        \Psi^{(i)}_{A\cup B}
    \right).
\]

We refrain from repeating his original proof herein, because 
in the rest of this section we provide a generalized form of this 
logical obstruction for superset-closed adversaries.
By the appropriate level of generalization, 
our presentation is much clearer and more succinct than the original one.

\subsection{Logical obstruction to adversarial set agreement task}
\label{subsec:obstAdversary}

We study the unsolvability of set agreement tasks under adversary schedulers, 
generalizing the basic wait-free scheduler to the ones
that allow nonuniform failures. 

We model such an adversary scheduler by a set $\advA$,
called an \keyword{adversary} \cite{GafniKuznetsov10}, 
of possible subsets of correct processes,
that is, $\advA$ is a nonempty subset of $\PowerSet{\colorSet}$
such that 
$\emptyset\not\in\advA$.
An adversary $\advA$ is called \keyword{superset-closed}, if
$P\in\advA$ and $P\subseteq P'\subseteq\colorSet$ implies $P'\in \advA$.
In this paper, we are solely concerned with the superset-closed adversaries
and a superset-closed adversary is simply called an adversary.  

An adversary $\advA$ has two different characterizations by
survivor sets and cores \cite{JunqueiraMarzullo07}.
An adversary $\advA$ can be specified by survivor sets, 
where a \keyword{survivor set} $S$ is a minimal set contained in $\advA$, 
that is, $S\supsetneq P$ for none of $P\in \advA$; 
Dually, an adversary $\advA$ can be specified by cores,
where a \keyword{core} $C$ is a minimal subset of $\colorSet$ 
satisfying $C\cap P\neq\emptyset$ for every $P\in\advA$.
Remark that the wait-free scheduler is just an instance of 
(uniform) adversary that is characterized by survivor sets
$\{ \{a\} \mid a\in\colorSet \}$.

We write $\csize(\advA)$ to denote the \keyword{minimum core size} 
of an adversary $\advA$, i.e., 
the minimum cardinality of the cores defined by 
$\csize(\advA)=\min \{\, \abs{C} \mid \text{$C$ is a core of $\advA$}\}$.

Following \cite{HerlihyRajsbaum10}, 
we model a single round 
execution of distributed processes against a given adversary $\advA$
by a protocol specified by the \keyword{round operator} $\Round{\advA}$. 
For each facet $X\in\smplInput{}$, which corresponds to an initial 
configuration of inputs to processes, 
the round operator $\Round{\advA}$ associates 
an output facet $\{ (a,\mathit{view}_a) \mid a\in\colorSet\}$,
where each $\mathit{view}_a$ is a subset of $X$
assigned to the $a$-colored process
and the set of views satisfy the following properties:
\begin{description}
    \item[(survival)] For each $a\in\colorSet$, 
    $\mathit{view}_a$ is a subset of $X$
    that subsumes a survivor set, i.e., $\coloring(\mathit{view}_a)\in \advA$;
    \item[(self-inclusion)] For each $a\in\colorSet$, $a\in \coloring(\mathit{view}_a)$; 
    \item[(containment)] Either $\mathit{view}_a \subseteq \mathit{view}_{a'}$
    or $\mathit{view}_a \supseteq \mathit{view}_{a'}$ holds, for every $a,a'\in\colorSet$.
\end{description}

Note that the wait-free immediate snapshot is 
further required to satisfy the immediacy condition:
For every $a,a'\in\colorSet$,
$\mathit{view}_{a'} \subseteq \mathit{view}_{a}$ holds 
whenever $a'\in \coloring(\mathit{view}_a)$. 
Missing the immediacy condition, 
the underlying simplicial structure of the round operator is more involved in
general. Even under the wait-free adversary, the round operator 
results in a non-manifold, as studied in \cite{BenavidesRajsbaum16}.
To cope with this extra complexity, 
our impossibility proof refines Nishida's original one.
Specifically we need to show an additional property, as stated in lemma~\ref{lem:minview}.

We define the round operator
as the product update model $\smplInput{}[\Round{\advA}]$,
where, in abuse of notation,   
$\Round{\advA}$ denotes 
the uniform communicative action model 
$\anglpair{\Facet(\Round{\advA}),\sim^{\Round{\advA}},\amPre{}}$
such that:
\begin{itemize}
    \item Each facet in $\Facet(\Round{\advA})$ is denoted 
    by $\actionRO{X}{\anglpair{S_0,\ldots,S_n}}$, 
    where $X\in\Facet(\smplInput{})$ and 
    the vector $\anglpair{S_0,\ldots,S_n}$ of subsets of $\colorSet$ 
    satisfies:
    \begin{itemize}
        \item
        for each $a\in\colorSet$, $S_{a}$ subsumes a survivor set, i.e., $S_{a}\in\advA$,
        \item 
        for each $a\in\colorSet$, $a\in S_{a}$, and
        \item 
        either $S_{a}\subseteq S_{a'}$ or $S_{a}\supseteq S_{a'}$ holds,
        for every $a,a'\in \colorSet$; 
    \end{itemize}    %
    \item The preconditions are given as in
    the definition of uniform action models in Section~\ref{subsec:taskDEL}.
\end{itemize}
In the following, we may occasionally write $\vec{S}$ to denote the vector 
$\anglpair{S_0,\ldots,S_n}$ for brevity. 

For a facet $\actionRO{X}{\anglpair{S_0,\ldots,S_n}}$, 
let us define a pair of functions 
$\InputAssig{\actionRO{X}{\anglpair{S_0,\ldots,S_n}}}$
and 
$\ViewAssig{\actionRO{X}{\anglpair{S_0,\ldots,S_n}}}$ by 
$\InputAssig{\actionRO{X}{\anglpair{S_0,\ldots,S_n}}}(a)=v$ if and only if 
$(a,v)\in X$ and also by
$\ViewAssig{\actionRO{X}{\anglpair{S_0,\ldots,S_n}}}(a)=S_a$, respectively.
Then, $\actionRO{X}{\vec{S}} \sim^{\smplInput{}[\Round{\advA}]}_a
\actionRO{Y}{\vec{S'}}$ holds iff
$\ViewAssig{\actionRO{X}{\vec{S}}}(a)=\ViewAssig{\actionRO{Y}{\vec{S'}}}(a)$
and 
$\InputAssig{\actionRO{X}{\vec{S}}}(a')=\InputAssig{\actionRO{Y}{\vec{S'}}}(a')$
for every $a'\in\ViewAssig{\actionRO{X}{\vec{S}}}(a)$.
In particular, 
$\actionRO{X}{\vec{S}} \sim^{\smplInput{}[\Round{\advA}]}_a
\actionRO{X}{\vec{S'}}$ iff 
$\ViewAssig{\actionRO{X}{\vec{S}}}(a)=\ViewAssig{\actionRO{X}{\vec{S'}}}(a)$.

Assuming an adversary $\advA$ is fixed, 
we define a class of positive epistemic formulas indexed by $A \in \PowerSet{\colorSet}$:
        \begin{align*}
        \Psi_{A}
        &\defeq
        \begin{cases}
            \Pfalse
            & \text{if no survivor set is subsumed by $\colorSet\setminus A$;}
            \\
            \ModD{A}\psi_{A}
            & \text{otherwise,}
        \end{cases}
    \end{align*} 
    where $\psi_{A}$'s are formulas given below 
    \begin{align*}
        \psi_{A}
        &\defeq
        \bigvee_{a\in \colorSet\setminus A}\neg\Pinput{a}{a} ~
        \vee \bigvee_{a\in \colorSet\setminus A}\ModK{a}
        \Biggl(\bigvee_{j\in A} \bigvee_{a'\in\colorSet}\Pinput{a'}{j}\Biggl)
        \vee \bigvee_{B\supsetneq A}\Psi_{B},
    \end{align*}
    defined by
    induction on $A$, ordered by inverse set inclusion
    (with $\psi_{\colorSet}$ being the base case of induction).

    In what follows, 
    we claim that the following formula is the logical obstruction
    against the adversary~$\advA$:
    \[
        \Phi
        \defeq
        \bigvee_{a\in\colorSet}\neg\Pinput{a}{a} ~
        \vee \bigvee_{A\in[\colorSet]^{<c}}\Psi_{A},
    \]
    where $c=\csize(\advA)$ and 
    $[\colorSet]^{<c}=\{A\subseteq\colorSet\mid 0<\abs{A}<c\}$.

Let us write $\minView{\actionRO{X}{\vec{S}}}$ to denote 
the minimum of the views of a facet $\actionRO{X}{\vec{S}}\in\Facet(\smplInput{}[\Round{\advA}])$, 
i.e., $\minView{\actionRO{X}{\vec{S}}}
= \bigcap_{a\in\colorSet} \ViewAssig{\actionRO{X}{\vec{S}}}(a)$.

\begin{lemma} \label{lem:minview}
    Suppose a facet $\actionRO{X}{\vec{S}}\in\Facet(\smplInput{}[\Round{\advA}])$ satisfies
    \begin{enumerate}[(i)]
        \item \label{lem:minview:i}
        $\smplInput{}[\Round{\advA}], \actionRO{X}{\vec{S}} 
        \models\bigwedge_{a\in\colorSet}\Pinput{a}{a}$, and
        \item \label{lem:minview:ii}
        for any $a\in\minView{\actionRO{X}{\vec{S}}}$, 
        $\ViewAssig{\actionRO{X}{\vec{S}}}(a)=\minView{\actionRO{X}{\vec{S}}}$.
    \end{enumerate}
    Then we have $\smplInput{}[\Round{\advA}], \actionRO{X}{\vec{S}} 
    \not\models\psi_{\colorSet\setminus\minView{\actionRO{X}{\vec{S}}}}$.
\end{lemma}

\begin{proof}
    Let us write $A$ for the minimum view $\minView{\actionRO{X}{\vec{S}}}$.
    Note that $A$ subsumes a survivor set. 

    We proceed the proof by induction on $A$.
    Let us consider the case that $A$ is exactly a survivor set. 
    For any $B$ that is a proper superset of $\colorSet\setminus A$, 
    $\colorSet\setminus B$ is a proper subset of $A$
    and hence contains no survivor set. Thus 
    by the assumption~\eqref{lem:minview:i},
    $\psi_{\colorSet\setminus A}$ is logically equivalent to
    $\xi_A \equiv 
        \bigvee_{a\in A}
        \ModK{a}\Bigl(
            \bigvee_{j\in \colorSet\setminus A}\bigvee_{a'\in\colorSet}\Pinput{a'}{j}\Bigr)
            $.
    We may suppose $A\neq\colorSet$, because 
    $\xi_A$ does not hold otherwise.    
    To see $\smplInput{}[\Round{\advA}], \actionRO{X}{\vec{S}}  \not\models \xi_A$, 
    it suffices to show that there is a facet $\actionRO{Y}{\vec{S'}}$
    satisfying $\actionRO{Y}{\vec{S'}}\sim_{a}\actionRO{X}{\vec{S}}$ for every $a\in A$ and also 
    $\InputAssig{\actionRO{Y}{\vec{S'}}}(a')\in A$ for every $a'\in\colorSet$.
    Take $k\in A$ arbitrarily. 
    Then such a facet is instantiated by $\actionRO{Y}{\anglpair{S_0',\ldots,S_n'}}$ 
    where $Y= \{(a,a)\mid a\in A\} \cup \{(a,k)\mid a\in \colorSet\setminus A\}$,
    $S_a' = S_a=A$ if $a\in A$, and $S_a'=\colorSet$ if $a\in \colorSet\setminus A$.

    Let us then consider the case $A$ is a proper superset of a survivor set.
    Similarly as above, we can prove
    $\smplInput{}[\Round{\advA}], \actionRO{X}{\vec{S}} \not\models 
        \bigvee_{i\in A}\neg\Pinput{i}{i} ~
        \vee\bigvee_{i\in A}\ModK{i}
        \Bigl(\bigvee_{j\in \colorSet\setminus A}\bigvee_{a'\in\colorSet}\Pinput{a'}{j}\Bigr)$. 
    Hence it suffices to show that $\smplInput{}[\Round{\advA}], \actionRO{X}{\vec{S}}
    \not\models \Psi_{\colorSet\setminus B}$ for every 
    proper subset $B$ of $A$.
    For such a proper subset $B$,
    let us define $\vec{S'}$ by 
    $S'_a = B$ if $a\in B$ and otherwise $S'_a= S_a$. 
    Then, the properties \eqref{lem:minview:i} and \eqref{lem:minview:ii}
    are held by the facet $\actionRO{X}{\vec{S'}}$
    and hence by induction hypothesis we have  
    $\smplInput{}[\Round{\advA}], \actionRO{X}{\vec{S'}}\not\models\psi_{\colorSet\setminus B}$.
    Since $\actionRO{X}{\vec{S}} \sim_{\ModD{\colorSet\setminus B}} \actionRO{X}{\vec{S'}}$, 
    we conclude that 
    $\smplInput{}[\Round{\advA}], \actionRO{X}{\vec{S}}\not\models
    \ModD{\colorSet\setminus B}\psi_{\colorSet\setminus B}$.
\end{proof}

\begin{proposition} \label{prop:nvalidRA}
    $\smplMf{I}[\Round{\advA}]\nvalid\Phi$.
\end{proposition}

\begin{proof}
    Let $\actionRO{X}{\vec{S}}$ be the facet of $\smplMf{I}[\Round{\advA}]$
    such that $X=\{(a,a)\mid a\in\colorSet\}$ and
    $S_a= \colorSet$ for every $a\in\colorSet$.
    We show $\smplMf{I}[\Round{\advA}],\actionRO{X}{\vec{S}}\not\models\Phi$.
    Clearly 
    $\smplMf{I}[\Round{\advA}],\actionRO{X}{\vec{S}}\not\models
    \bigvee_{a\in\colorSet}\neg\Pinput{a}{a}$
    and thus it suffices to show $\smplMf{I}[\Round{\advA}],\actionRO{X}{\vec{S}}\not\models
    \bigvee_{A\in[\colorSet]^{<c}}\Psi_{A}$.
    Take any $A\in [\colorSet]^{<c}$ arbitrarily.
    Since $\abs{A}<c$, $A$ does not subsume a core and thus 
    $A \cap S = \emptyset$ for some survivor set $S\in\advA$. 
    This survivor set $S$ is subsumed by  $\colorSet\setminus A$
    and hence 
    we are obliged to show $\smplMf{I}[\Round{\advA}],\actionRO{X}{\vec{S}} \models 
    \ModD{A}\psi_A$.
        
    Let us define $\vec{S'}$ by $S'_a = \colorSet\setminus A$ if $a\in \colorSet\setminus A$
    and $S'_a = \colorSet$ otherwise. 
    Then 
    $\smplMf{I}[\Round{\advA}],\actionRO{X}{\vec{S'}}\not\models
    \psi_A$ follows from 
    lemma~\ref{lem:minview}. Since $\actionRO{X}{\vec{S'}}\sim_{\ModD{A}}\actionRO{X}{\vec{S}}$, 
    we have $\smplMf{I}[\Round{\advA}],\actionRO{X}{\vec{S}}\not\models
    \Psi_A$. As we have taken 
    $A\in [\colorSet]^{<c}$ arbitrarily, 
    this concludes $\smplMf{I}[\Round{\advA}],\actionRO{X}{\vec{S}}\not\models
    \bigvee_{A\in[\colorSet]^{<c}}\Psi_{A}$.
\end{proof}

\begin{lemma} \label{lem:knownInput}
    Suppose $X\times\vec{d}$ is a facet of $\cplI[\amSA{k}]$.
    Then, 
    for every
    $a\in\colorSet$ and $A\subseteq\colorSet$, 
    $\OutputAssig{X\times\vec{d}}(a)\in A$ implies
    $\cplI[\amSA{k}], X\times\vec{d}\models
    \ModK{a}\bigl(\bigvee_{i\in A}\bigvee_{a'\in\colorSet}\Pinput{a'}{i}\bigr)$.
\end{lemma}

\begin{proof}
    Let $Y\times\vec{d'}$ be any facet such that 
    $X\times\vec{d}\sim_{a}^{\cplI[\amSA{k}]} Y\times\vec{d'}$.
    Then, $\OutputAssig{Y\times\vec{d'}}(a)=\OutputAssig{X\times\vec{d}}(a)\in A$.
    Since $\InputAssig{Y\times\vec{d'}}(\colorSet)\supseteq 
    \OutputAssig{Y\times\vec{d'}}(\colorSet)$,
    it follows that $\InputAssig{Y\times\vec{d'}}(a')\in A$
    for some $a'\in\colorSet$. This implies
    $\cplI[\amSA{k}], Y\times\vec{d'}\models
    \bigvee_{i\in A}\bigvee_{a'\in\colorSet}\Pinput{a'}{i}$.
    Hence we are done.
\end{proof}

\begin{proposition} \label{prop:validSAc}
    Let $\advA$ be any adversary. 
    Then, 
    $\smplInput{}[\amSA{k}]\valid\Phi$ holds for every $k<\csize(\advA)$.
\end{proposition}

\begin{proof}
    Let $k<\csize(\advA)$ and 
    $X\times\vec{d}$ be any facet of $\smplInput{}[\amSA{k}]$. 
    We will show $\smplInput{}[\amSA{k}], X\times\vec{d}\models\Phi$.
    We may suppose $\smplInput{}[\amSA{k}], X\times\vec{d}\not\models
    \bigvee_{a\in\colorSet}\neg \Pinput{a}{a}$, 
    because $\smplInput{}[\amSA{k}], X\times\vec{d}\models\Phi$ immediately 
    holds otherwise.
    By lemma~\ref{lem:existPermsubset}, the function $\OutputAssig{X\times\vec{d}}$
    has 
    a permutation subset $A$. 
    Since $\abs{A}\leq \max \{ \abs{\OutputAssig{X\times\vec{d}}(Z)} \mid 
    Z\subseteq\colorSet\}\leq k< \csize(\advA)$, 
    it suffices to show that 
    $\smplInput{}[\amSA{k}], X\times\vec{d}\models\Psi_{A}$
    for this permutation subset $A$. 
    Since $\abs{A}<\csize(\advA)$, by a similar discussion in the proof
    of Proposition~\ref{prop:nvalidRA}, 
    we are obliged to show $\smplInput{}[\amSA{k}], X\times\vec{d} \models 
    \ModD{A}\psi_A$.
    
    We proceed by induction 
    on the size of $\colorSet\setminus A$.
    First, consider the base case $\abs{A}=k$.
    Let us show that $\smplInput{}[\amSA{k}], Y\times\vec{d'}\models\psi_{A}$ holds
    for any facet $Y\times\vec{d'}$ such that
    $Y\times\vec{d'}\sim_{\ModD{A}}^{\smplInput{}[\amSA{k}]}X\times\vec{d}$.
    Since $A=\OutputAssig{X\times\vec{d}}(A)=\OutputAssig{Y\times\vec{d'}}(A)
    \subseteq \OutputAssig{Y\times\vec{d'}}(\colorSet)$
    and $\abs{\OutputAssig{Y\times\vec{d'}}(\colorSet)}\leq k$,
    we have $\OutputAssig{Y\times\vec{d'}}(\colorSet)\subseteq A$.
    Hence by lemma~\ref{lem:knownInput},
    $\smplInput{}[\amSA{k}],Y\times\vec{d'}
    \models\ModK{a}\bigl(\bigvee_{i\in A}\bigvee_{a'\in\colorSet}\Pinput{a'}{i}\bigr)$
    holds for every $a\in\colorSet$.
    This entails $\smplInput{}[\amSA{k}],Y\times\vec{d'}\models\psi_{A}$.

    Next, consider the case $\abs{A}<k$. 
    Again,
    let us show that $\smplInput{}[\amSA{k}], Y\times\vec{d'}\models\psi_{A}$
    holds for any facet $Y\times\vec{d'}$ such that 
    $Y\times\vec{d'}\sim_{\ModD{A}}^{\smplInput{}[\amSA{k}]}X\times\vec{d}$.
    Since $\OutputAssig{X\times\vec{d}}$ has $A$ as a permutation subset, 
    so does $\OutputAssig{Y\times\vec{d'}}$. 
    We will show that $\smplInput{}[\amSA{k}], Y\times\vec{d'}\models
    \bigvee_{B\supsetneq A}\Psi_{B}$ holds, 
    assuming 
    $\smplInput{}[\amSA{k}],Y\times\vec{d'}\not\models
    \bigvee_{a\in \colorSet\setminus A}\neg\Pinput{a}{a} ~
    \vee \bigvee_{a\in \colorSet\setminus A}\ModK{a}
    \bigl(\bigvee_{j\in A}\bigvee_{a'\in\colorSet}\Pinput{a'}{j}\bigr)$.
    It follows from
    $\smplInput{}[\amSA{k}],Y\times\vec{d'}\not\models 
    \bigvee_{a\in \colorSet\setminus A}\ModK{a}
    \bigl(\bigvee_{j\in A}\bigvee_{a'\in\colorSet}\Pinput{a'}{j}\bigr)$
    that $\OutputAssig{Y\times\vec{d'}}(\colorSet\setminus A)
    \subseteq \colorSet\setminus A$ by lemma~\ref{lem:knownInput}.
    Thus by lemma~\ref{lem:existPermsubset}, 
    $\OutputAssig{Y\times\vec{d'}}$ has a permutation subset $A'$,
    that is, $\OutputAssig{Y\times\vec{d'}}(A')=A'\subseteq \colorSet\setminus A$.
    Therefore $A\cup A'$ is a
    proper superset of $A$ and is also a permutation subset for 
    $\OutputAssig{Y\times\vec{d'}}$.  
    By induction hypothesis we have 
    $\smplInput{}[\amSA{k}],Y\times\vec{d'}\models \Psi_{A\cup A'}$,
    which entails $\smplInput{}[\amSA{k}],Y\times\vec{d'}\models \bigvee_{B\supsetneq A}\Psi_{B}$.
\end{proof}

By proposition~\ref{prop:nvalidRA} and \ref{prop:validSAc},
$\Phi$ is a logical obstruction and hence the 
following impossibility result follows.
\begin{theorem} \label{thm:impossibilitySA}
    Let $\advA$ be an adversary.  If $k<\csize(\advA)$, 
    the $k$-set agreement task $\smplInput{}[\amSA{k}]$
    is not solvable the protocol 
    $\smplInput{}[\Round{\advA}]$,
    a single round of the round operator.
\end{theorem}


\section{Conclusion and Future Work}
\label{sec:conclusion}

We have applied the logical method 
developed by Goubault, Ledent, and Rajsbaum \cite{GoubaultLedentRajsbaum:GandALF18}
to show the impossibility of the set agreement tasks for 
superset-closed adversaries,
by giving concrete logical obstruction formulas, which are
generalization of Nishida's logical obstruction for the wait-free model \cite{Nishida:Msc20}.
The method based on logical obstruction allows an
elementary inductive proof, without recourse to 
sophisticated topological tools.
The instances of logical obstruction 
exemplify that logical method
would serve as a feasible alternative to topological method.

There are several topics to pursue that merit further investigation.
%
First, it is quite interesting what varieties of impossibility
can be proven by devising concrete logical obstruction formulas. 
To date, the topological method has been extensively applied to show 
impossibility results 
and still keep expanding its application area, 
e.g., network computing \cite{CastanedaFPRRT19,FraigniaudPaz20}. 
Finding concrete logical obstruction formulas to
these particular instances would be an interesting topic of its own right. 

We would also expect that logical obstruction could 
give a deeper logical understanding on the nature of unsolvability
of distributed tasks.
As observed in \cite{GoubaultLazicLedentRajsbaum:DaLi19}, 
the original logical framework of \cite{GoubaultLedentRajsbaum:GandALF18}, 
whose atomic propositions are allowed to mention input values only, 
cannot refute the solvability of the equality negation task \cite{LoHadzilacos:00},
while an extended framework that allows atomic propositions to mention 
output decision values has logical obstruction.
This seems to suggest that solvability of distributed tasks could be
classified in terms of expressibility of epistemic logic.


As we have mentioned in Introduction,
this paper solely concerns single round protocols,
leaving multiple round protocols for future investigation.  
In principle we are able to discuss multiple round protocols in the epistemic
logic setting, but we would need
to work on a different epistemic model for each incremental round step.
It would be a challenging topic to give a remedy for this
by providing, say, a series of obstruction formulas 
indexed by the number of round steps.

\subsection*{Acknowledgment}
The second author would like to thank Yutaro Nishida, 
who left academia just after writing up his Master's thesis,
for fruitful discussion on the topic. 
The second author is supported by JSPS KAKENHI Grant Number 20K11678.

{
\bibliographystyle{plainurl}
\bibliography{distrib}
}

\end{document}